\documentclass[pra,twocolumn,showpacs,superscriptaddress]{revtex4}
\usepackage{amssymb}
\usepackage{amsmath}
\usepackage{graphicx}
\usepackage{dcolumn}
\usepackage{bm}

\newtheorem{theorem}{Theorem}[section]

\newenvironment{proof}[1][Proof]{\begin{trivlist}
\item[\hskip \labelsep {\bfseries #1}]}{\end{trivlist}}

\begin{document}

\preprint{APS/123-QED}
\title{ The Essentially Entangled Component of Multipartite  Mixed Quantum States, its Properties and an Efficient Algorithm for
its Extraction}
\author{V.\ M.\ Akulin}
\affiliation{Laboratoire Aim\'{e} Cotton, CNRS (UPR 3321), B\^{a}timent 505, 91405 Orsay
Cedex, France. }
\affiliation{Institute for Information Transmission Problems of the Russian Academy of
Science, Bolshoy Karetny per. 19, Moscow, 127994, Russia.}
\author{G. A. Kabatyanski}
\affiliation{Institute for Information Transmission Problems of the Russian Academy of
Science, Bolshoy Karetny per. 19, Moscow, 127994, Russia.}
\affiliation{Laboratoire J.-V. Poncelet CNRS (UMI 2615) Bolshoi Vlassievsky per. 11,
Moscow, 119002 Russia.}
\author{A.\ Mandilara}
\affiliation{Department of Physics, School of Science and Technology, Nazarbayev University, 53, Kabanbay
batyr Av., Astana, 010000, Republic of Kazakhstan.}


\begin{abstract}
We introduce with geometric means  a  density matrix decomposition  of a multipartite quantum  system of a finite dimension into
two density matrices: a separable one, also known as the best separable approximation,  and an essentially entangled one, which  contains no product states components.  
We show that this convex decomposition solving the separability problem, can be achieved in practice with the help of an algorithm  based on linear programming, which in the general case scales polynomially  with the dimension of the multipartite system. Furthermore, we suggest methods for analyzing the multipartite entanglement content of the essentially entangled component and derive analytically an upper bound for its rank.
 We illustrate the algorithm at an example of a composed system of total dimension $12$ undergoing loss of coherence due to classical noise and we trace the time evolution of its essentially entangled component. We suggest a ``geometric'' description of entanglement dynamics and show how it explains the well-known phenomena  of sudden death and revival of multipartite entanglement. 

\end{abstract}

\maketitle

\section{Introduction}

Though quantum entanglement  is a concept which has attracted much of the attention of physicists working in
various fields \cite{Amico},  still,  there  remains room for further progress on its understanding \cite{Horo}. One of the main open problems is the efficient detection and characterization of multipartite entanglement  of  density matrices representing open quantum systems undergoing non-unitary evolution \cite{Davidovich}.

All experimentally addressable information about a
quantum physical system is summarized
in its density matrix $\widehat{\rho }$.
We focus on a  multipartite quantum system, which comprises a finite
number $K<\infty $ of parts $\mathcal{N}_{k}$ numerated by index $k=1,\ldots
,K$, each of which has the Hilbert space of a finite dimensionality $N_{k}$,
whence $\prod_{k=1}^{K}N_{k}=N$ is the dimensionality of the Hilbert space
of the entire system. 
This system --assembly of parts,
is called entangled (or inseparable) if and only if its
density matrix cannot be caste as a statistical sum
\begin{equation}
\widehat{\rho }\neq \sum_{i=1}^{M}a_{i}\prod\limits_{\otimes
k=1}^{K}\left\vert \alpha _{i}^{k}\right\rangle \left\langle \alpha
_{i}^{k}\right\vert ,  \label{EQ0}
\end{equation}%
 ($a_{i}>0$, $\sum_{i=1}^{M}a_{i}=1$) of $
M$ various ($i=1,\ldots ,M$) direct products $\prod\limits_{\otimes
k=1}^{K}\left\vert \alpha _{i}^{k}\right\rangle \left\langle \alpha
_{i}^{k}\right\vert $ of the density matrices $\left\vert \alpha
_{i}^{k}\right\rangle \left\langle \alpha _{i}^{k}\right\vert $ of pure
states $\left\vert \alpha _{i}^{k}\right\rangle $ of each part. This
condition provides the most general case of entangled systems opposite to a
separable quantum system comprised of statistically independent elements,
where Eq.(\ref{EQ0}) holds as an equality. 

Many approaches \cite{Horo} have been developed so far aiming to answer the question whether or not a density matrix is separable. 
Concerning exact analytic results, up to now, there is no  method applicable  to the multipartite problem, and we  believe that such a solution does not exist at all. 
An algorithmic  solution to the ``decision'' problem \cite{NP} associated with separability has been conjectured to be a NP hard problem 
but valuable progress has been done (mainly on the bi-separability problem) in approaches  \cite{Doherty}-\cite{Ioannou2} where semidefinite programming is merged with analytic criteria \cite{Peres}. 

In this work we provide a geometric point of view on the problem of inseparability  that suggests an efficient solution based on linear programming. Employing simple geometric arguments we  suggest an algorithm that results to a \textit{unique} decomposition of   the density matrix  as \begin{equation}\widehat{\rho }= (1-B) \widehat{\rho }_{\mathrm{sep}}+ B \widehat{\rho }_{\mathrm{ent}}\label{one}\end{equation} where $\widehat{\rho }_{\mathrm{sep}}$  is, what we call in this work, the  \textsl{separable component}, $\widehat{\rho }_{\mathrm{ent}}$ the 
\textsl{essentially entangled} part  which cannot have any separable states as components and $B$ is a positive number in the range $\left[0,1\right]$. Obviously, the  decomposition, Eq.(\ref{one}), implies that the state $\widehat{\rho }$ is  separable in all $K$ parts only for $B=0$.

The decomposition in Eq.(\ref{one}) has been initially introduced in \cite{Sanpera} without resorting to a geometric picture and the component $(1-B) \widehat{\rho }_{\mathrm{sep}}$  is widely known as \textit{the best separable approximation} of the  density matrix $\widehat{\rho }$. In that same seminal work, the uniqueness of the decomposition has been proven for  the  multipartite case and  a strict upper bound on the rank 
of the component $  \widehat{\rho }_{\mathrm{ent}}$ for the biseparable case. In this work we generalize the latter to the multipartite case, proving that the rank of  $\widehat{\rho }_{\mathrm{ent}}$ is upper bounded by a number related to the dimensions of the total system and those of the sub-elements.

On a practical level, we  show that linear programming algorithm combined with a simple optimization technique allows one
to efficiently find the decomposition of a generic density matrix
\begin{equation}
\widehat{\rho }=\sum_{i=1}^{M}a_{i}\underset{\mathrm{product\ states}}{\underbrace{\prod\limits_{\otimes k=1}^{K}\left\vert \alpha
_{i}^{k}\right\rangle  \left\langle \alpha _{i}^{k}\right\vert }}+\sum_{i=1}^{m}b_{i}\underset{\mathrm{entangled\ states}}{\underbrace{\left\vert \beta _{i}\right\rangle \left\langle \beta _{i}\right\vert}} ,
\label{EQ1}
\end{equation}
with the coefficients  constrained by the requirements 
\begin{eqnarray}
a_{i}>0,\quad b_{i}\geq 0,\quad
\sum_{i=1}^{M}a_{i}+\sum_{i=1}^{m}b_{i}=1,\\ \quad \mathrm{and}\quad \sum_{i=1}^{m}b_{i}\rightarrow\min.  \label{EQ2}
\end{eqnarray}
 When this limit is reached, the decomposition in Eq.(\ref{EQ1}) yields Eq.(\ref{one}) with $B=\left(\sum_{i=1}^{m}b_{i}\right)_{\min} $:
\begin{equation}
\widehat{\rho }_{\mathrm{sep}}=\sum_{i=1}^{M}\frac{a_{i}}{1-B}\underset{\mathrm{product\ states}}{\underbrace{\prod\limits_{\otimes k=1}^{K}\left\vert \alpha
_{i}^{k}\right\rangle  \left\langle \alpha _{i}^{k}\right\vert }}\label{sep},
\end{equation}
and
\begin{equation}
\widehat{\rho }_{\mathrm{ent}}=\sum_{i=1}^{m}\frac{b_{i}}{B}\underset{\mathrm{entangled\ states}}{\underbrace{\left\vert \beta _{i}\right\rangle \left\langle \beta _{i}\right\vert}}.\label{ent}
\end{equation}

 Is known that  the linear programming method, in the general case, scales polynomially with the dimension of the vector space where it is applied. Employing the fact that $M+m\leq N^2$ in Eq.(\ref{EQ1}), where $N$ is  dimension of quantum assembly  under study, we show that the proposed algorithm yielding the decomposition Eq.(\ref{EQ1}) scales   as  $(2N^4)^3$.  
 
The paper is structured as follows. In Section~\ref{A} we introduce  the idea of the
decomposition Eq.(\ref{one}) and illustrate its uniqueness  with a simple geometric picture generalizing the Bloch vector representation of a
two-level system. This picture also helps us to analyze some properties of   $\widehat{\rho }_{\mathrm{ent}}$  and we conclude this section with a theorem setting an upper limit on its  rank. In Section~\ref{B} we
present a version of an efficient linear programming algorithm allowing one
to explicitly find the decomposition Eqs.(\ref{EQ1})-(\ref{EQ2}) for a generic density matrix.
In Section~\ref{C}  we
suggest ideas for characterizing entanglement of the component   $\widehat{\rho }_{\mathrm{ent}}$ which naturally
reflects the entanglement properties of $\widehat{\rho }$.
 In Section~\ref{D} we present a physical example  which
demonstrates the  implementation of the technique introduced in previous sections and connects them with known notions in open quantum system dynamics. We conclude by the discussion in  Section~\ref{E}.

\section{The geometric idea of  decomposition  and  properties of the essentially entangled part}\label{A}

All possible density matrices of a  quantum system with a Hilbert space of  dimension $N$, are comprising a  convex
set of positive Hermitian matrices of unit trace. This set can be
viewed  as a manifold in the vector space of Hermitian matrices endowed with the  a metric -- the
Hilbert-Schmidt inner product $\mathrm{tr}\left[\widehat{\rho}_i,\widehat{\rho}_j\right]$. 
The requirement of the unit trace in this representation means that the inner product of a
vector representing a density matrix and a vector representing the unit
matrix equals to unity. Henceforth we  call this manifold ``Liouville vector space''. Furthermore, the density matrix of a pure state has
rank one, which implies that the length of the vector corresponding to a pure state,  equals to unity.
The density matrix manifold is thus a convex hull at the unit-length length vectors  having   unit projection on  the unity matrix. 
\begin{center}
 \begin{figure*}[t] \makebox[\textwidth][c]{\includegraphics[width=0.8\textwidth]{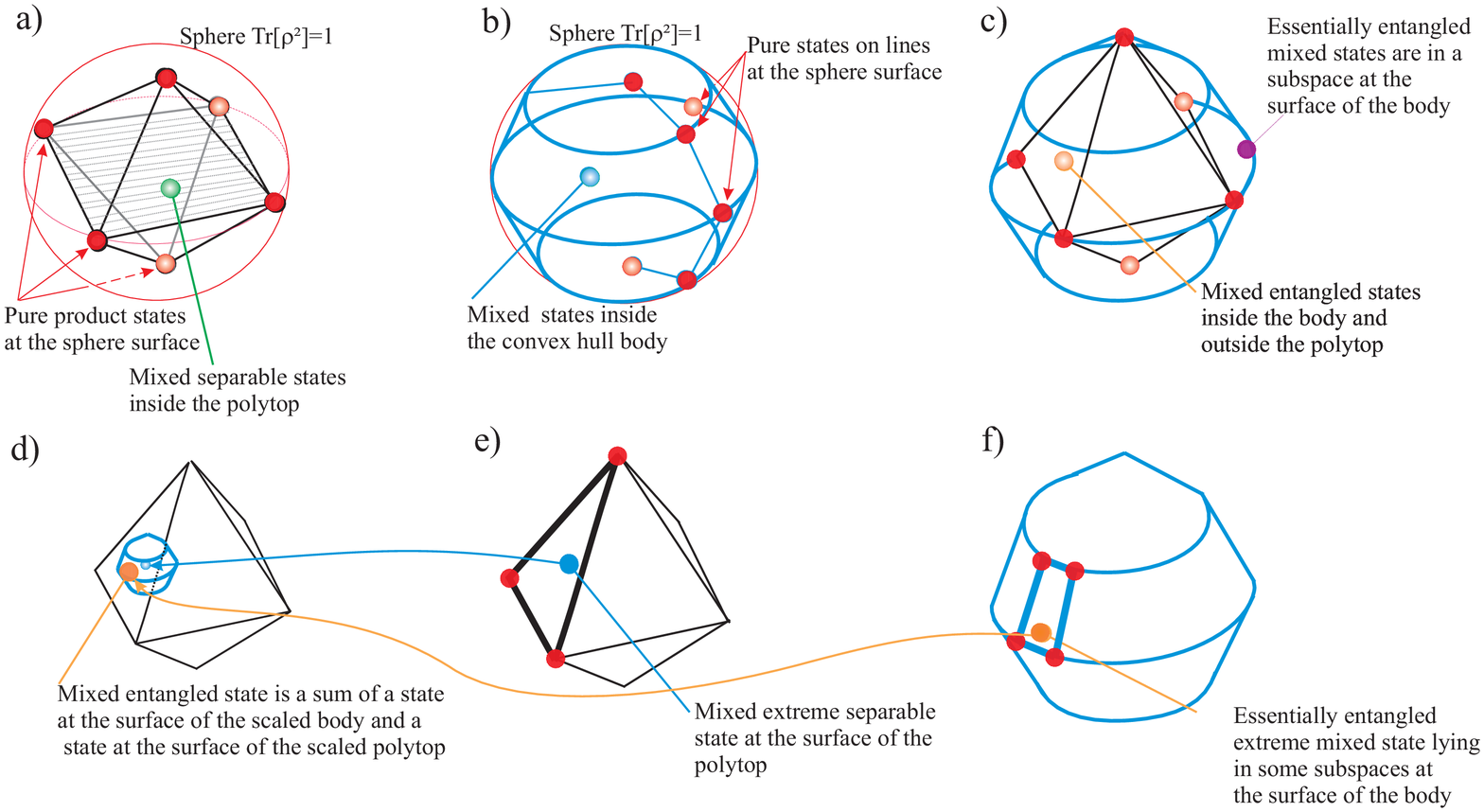}}
\vspace{2cm}
\caption{A symbolic illustration of the geometric structure of density matrices and of the decomposition Eq.(\ref{one}).}
 \label{FIG1}\end{figure*}
\end{center}

A natural basis exists for such a vector space suggested by the $N^{2}$
properly normalized generators $\widehat{g}_{i}^{N}$ of the unitary group $SU(N)$, 
including the unity $\widehat{I}=\sqrt{N}\widehat{g}_{0}^{N}$. This basis
allows one to cast a $N\times N$ density matrix  of a quantum system
as $\widehat{\rho }=\sum_{i=0}^{N^{2}-1}\widehat{g}_{i}^{N}r_{i}$ with  $r_{i}=\mathrm{Tr}[\widehat{g}_{i}^{N},\widehat{\rho }]$ the $N^2$ real vector components. This geometric picture is in direct analogy to the Bloch vector for  two-level systems.  The pure quantum states lay at the surface of the 
unit hypersphere, $\mathrm{Tr}[\widehat{\rho }^{2}]=\sum_{i=0}^{N^{2}-1}r_{i}^{2}=1$,  but
in contrast to the Bloch vector of $2$--dimensional pure quantum states, these states do not cover all the surface of the  hypersphere of dimension $N^2-2$
but are  confined at a manifold of lower dimensionality, $N^{2}-N$. This can be easily understood when the characteristic polynomial $\mathrm{Det}[\lambda -\widehat{\rho }]=\lambda ^{N}+c_{1}(\left\{r_{i}\right\} )\lambda ^{N-1}+c_{2}(\left\{r_{i}\right\} )\lambda ^{N-2}+\ldots +c_{N}(\left\{ r_{i}\right\} )$ of
a pure state is considered. The unit trace requirement ensures that $c_{1}(\left\{ r_{i}\right\} )\equiv - 1$, while the rank $1$ requirement implies the constraints  $c_{m}(\left\{ r_{i}\right\} )=0$ for $m=2,\ldots,N$. 

The set of $N$ conditions on the $N^2$ components of the vector representing a pure state, constrains the vector to lay on a restricted manifold  of lower dimension $\left(N^{2}-N\right) $ at the surface of the unit hypersphere.  As a  consequence, the density matrices for  quantum systems of dimension $N>2$ do not `fill' the whole inner part of the unit hypersphere,  but they are laying inside an $\left(N^{2}-2\right) $-dimensional body formed as a convex hull of the pure
states of the $\left(N^{2}-N\right) $ dimensional manifold. This convex hull   plays the role of the Bloch ball for higher dimensions of the Hilbert space and  has been exhaustively studied  in \cite{Caves} for the case of $3$-dimensional systems.  The convex hull    is touching the unit hypersphere
only for the  pure states while  its  outer hypersurface, which we denote by $S_{CH}$, is naturally  the border between positive and non-positive Hermitian matrices of unit trace. Therefore $S_{CH}$ 
consists only of the degenerate density matrices which have at least one zero eigenvalue. In Fig.\ref{FIG1}~(b) we
symbolically illustrate the convex hull of pure states, such that all
density matrices are inside this body.  

The situation is similar for the convex hull formed exclusively by the pure product  states. However   since the product states is a manifold of measure zero in the set of all states,   the
convex hull of pure product states is located inside the convex hull of all
 pure states, apart from the points at the unit hypersphere corresponding to the pure product states.  At the same time,
the outer surface of the convex hull of product states does not separate positive from
negative matrices, and hence it must not exclusively contain degenerate matrices. In Fig.\ref{FIG1}~(a) we illustrate the situation symbolically
by showing pure product states as points at the spherical surface and the
convex hull of these points by a polytope inside the sphere. At the surface
and inside the polytope the states are separable. 

In Fig.\ref{FIG1}~(c) we illustrate that inseparable states are the mixed states
inside the body symbolizing the convex hull of pure states but are outside
the polytope symbolizing the convex hull of the product states. In Fig.\ref{FIG1}~(e)-(f) we illustrate the geometric meaning of Eq.(\ref{EQ1}), that each mixed state can be represented as a sum of separable state symbolized by the
polytope within a scaled sphere of radius $a=\sum_{i=1}^{M}a_{i}$ and an
entangled state in the body within a scaled sphere of radius $b=\sum_{i=1}^{m}b_{i}=1-a$. In the situation where $b$ is minimum ($b=B$, Eq.(\ref{one})), the
corresponding extremum states are at the surfaces of the polytope and the
body, respectively, as shown in Figs.\ref{FIG1}~(e)-(f). Obviously for a given a state such a decomposition is unique and 
the extremum state on the outer surface $S_{CH}$ of the convex hull corresponds to essentially entangled component in Eq.(\ref{one}).

Let us now turn to the properties of the essentially entangled component $\widehat{\rho }_{\mathrm{ent}}$ which, as it will be shown now,
 is a density matrix of rank $d_E$ strictly less than  the dimension $N$ of the Hilbert space  of the entire system. The essentially
 entangled component belongs to the outer hypersurface $S_{CH}$ of the convex hull of all states,  but not every state on $S_{CH}$ 
is an essentially entangled component; only some of them which do not contain the separable part (see Fig.\ref{FIG1}~(c)). In addition, the eigenvectors of $\widehat{\rho}_{\mathrm{ent}}$, $\left|\psi_{l}\right\rangle$ of $\widehat{\rho}_{\mathrm{ent}}$ with $l=1,\ldots,d_E$, 
  are necessarily $K$\textit{-entangled} pure states in the sense that these cannot be written as direct product of $K$ pure states corresponding respectively to the $K$ subsystems. Henceforth, we call pure states which are direct products of $K$ pure states of the $K$ subsystem, $K$\textit{-product} states.

Consider now the Hilbert space $H_E$ of  dimension $d_E$, which is associated with the eigenvectors  $\left|\psi_{l}\right\rangle$ of 
$\widehat{\rho }_{\mathrm{ent}}$. Each state $\left|\bar{\psi}\right\rangle$ belonging to the Hilbert space $H_E$ is apparently a linear combination of the eigenvectors,
\begin{equation}
\left|\bar{\psi}\right\rangle=\sum_{l=1}^{d_E}\lambda_{l} \left|\psi_{l}\right\rangle \,. \;  \label{t1}
\end{equation}
The vector $\left|\bar{\psi}\right\rangle$  can be also seen as a result of the action of an element $\hat{U}_E$ of the unitary group $SU(d_E)$ associated with the Hilbert subspace $H_E$ at one of the eigenvectors, 
\begin{equation}\left|\bar{\psi}\right\rangle=\hat{U}_E \left|\psi_{1}\right\rangle \,.\; \label{t2}
\end{equation}
Now, let us consider the convex hull of the states $\left|\bar{\psi}\right\rangle$ of the subspace, which naturally contains $\widehat{\rho }_{\mathrm{ent}}$. The condition that $\widehat{\rho }_{\mathrm{ent}}$ does not have any separable components, $\left|\psi_{\mathrm{prod}}\right\rangle\left\langle \psi_{\mathrm{prod}}\right|$,  implies that the  convex hull does not contain a product state $\left|\psi_{\mathrm{prod}}\right\rangle\left\langle \psi_{\mathrm{prod}}\right|$ which is possible only if the Hilbert space $H_E$ does not contain $\left|\psi_{\mathrm{prod}}\right\rangle$. We name  a Hilbert subspace with such a property an \textit{essentially entangled subspace} of dimension $d_E$ and in what follows, with the help of this necessary condition,  we find an upper bound on $d_E$.

\begin{theorem}
The maximum rank $d_{E\max }$ of an essentially entangled component $\widehat{\rho }_{\mathrm{ent}}$ for a system of dimension $N$ comprised by $K$ subsystems each of  them of dimension $N_k$, is $N-\sum_{k=1}^KN_k+K-1$. 
\end{theorem}

\begin{proof}
Let us assume that the essentially entangled component $\widehat{\rho }_{\mathrm{ent}}$  is a density matrix of rank $d_E$,  and consider the subspace $H_E$ which is spanned by its $K$-entangled eigenvectors $\left\{\left|\psi_1\right\rangle,\left|\psi_2\right\rangle,\ldots\left|\psi_{d_E}\right\rangle\right\}$. Let us also consider the orthogonal compliment of the subspace $H_E$, $H_E^\bot$  of dimension $N-N_{d_E}$ and arbitrary  select a set of mutually orthogonal vectors  $\left\{\left|\chi_1\right\rangle,\left|\chi_2\right\rangle,\ldots\left|\chi_{N-d_E}\right\rangle\right\}$ spanning  $H_E^\bot$. 

The subspace $H_E$ is not essentially entangled, if there is at least one  product state $\left|\psi_{\mathrm{prod}}\right\rangle$ which can be expressed  as in Eq.(\ref{t1}),
\begin{equation}
\left|\psi_{\mathrm{prod}}\right\rangle=\sum_{l=1}^{d_E}\lambda_{l} \left|\psi_{l}\right\rangle \,. \; \label{gs}
\end{equation}
where $\lambda$'s are complex numbers. Equation (\ref{gs}) implies that  $\left|\psi_{\mathrm{prod}}\right\rangle$  must be orthogonal to  every element  $\left\{\left|\chi_i\right\rangle\right\}$, with $i=1,\ldots , N-d_E$, of the chosen basis in $H_E^\bot$, 
\begin{equation}
\left\langle \psi_{\mathrm{prod}}\right|\left.\chi_{i=1, \ldots , N-d_E}\right\rangle=0 \;.
\end{equation}
The maximum number of such conditions  equals to the number of parameters defining a product state, which for a $K$-product state amounts to $\sum_{k=1}^KN_k -K$. Therefore the maximum rank of an essentially $K$-entangled density matrix cannot be equal or exceed $\sum_{k=1}^KN_k -K$.

The maximum rank is smaller when we speak not about the essentially $K$-entangled component, but about  the essentially entangled component which does not contain, not only $K$-product, but any product state. For this case one has to identify the bi-partition of the system that yields the maximum number of parameters characterizing the product state.

In the Appendix we provide a more detailed proof of this theorem. 

\end{proof}
 If this theorem is applied to  the case of $2$ qubits in mixed state,  $d_{Emax}=1$ is obtained meaning that   that the essentially entangled component can only be a  pure entangled state. This result is in agreement with the results  in \cite{Sanpera} where the bipartite case is treated.  We would like to note, that the example studied in Section~\ref{D}  gives some preliminary evidence that    $\widehat{\rho }_{\mathrm{ent}}$ is prone to stay very near to pure states ($Tr\left[\widehat{\rho }_{\mathrm{ent}}^2\right]\approx1$) even though $d_{Emax}\rightarrow N$ for $N>>1$.

\section{The linear programming iteration algorithm that yields the
essentially entangled component of a density matrix}\label{B}

One can find the maximum separable and the essentially entangled components of an
arbitrary density matrix straightforwardly with the help of the linear
programming algorithm applied to the convex hull  of general pure states and the ``polytope'' of
pure separable quantum states. The main obstacle on this way is a high dimensionality of the corresponding Liouville
vector space, which makes intractable the direct approach within any
approximation. In fact, even for the simplest multipartite system of three
qubits, the dimensionality ($N^{2}$) of the density matrix space is $64$, such
that even for the rather low-accuracy approximation attributing just $10$
points per dimension, one encounters a polytope of already $10^{64}$
vertices.

Here, we suggest  a way to   crucially decrease the number of the vertices that enter as 
samples in the algorithm and, in consequence, the computational complexity of the procedure. We first notice the fact that the solution of the problem 
and, in general, any convex decomposition of the form Eq.(\ref{EQ1}), allows for at most $N^{2}$ non-zero coefficients $a_{i}$ and $b_{i}$. This observation can be formally justified by a theorem of Carath\'{e}odory as mentioned in \cite{Ioannou}.
In the limit  $B=\left(\sum_{i=1}^{m}b_{i}\right)_{\min}$ the  pure states are the vertices
associated with the corners of the facets corresponding to the solution, as illustrated
 in Fig.~\ref{FIG1}~(e)-(f), while other vertices can be discarded.

Therefore, at first step  we may randomly take $N^{4}$ product states, $N^{4}$ general states and in order to ensure the algorithmic stability, complement this set by the $N^2$ eigenvectors of the given density matrix. We find the solution of the linear programming problem, which
typically has complexity $\sim \left( 2N^{4}\right) ^{3}$, and thereby
identify at most $N^{2}-J$ product states and $J$ general states with
nonzero coefficients $a_{i}$ and $b_{i}$, respectively. The linear constraint imposed on the algorithm is the minimization of $\sum_{i=1}^{m}b_{i}$ and the solution provided is a `local' minimum, for the given set of vectors fed to the algorithm.  Our  aim is to find the global minimum value of  $\sum_{i=1}^{m}b_{i}$ that is equal to $B$ and to this end we create an iterative  optimization loop which guides us there. 


At the second and subsequent steps, we take the product states resulting from the solution of the optimization problem at the former step
and   by applying to each of them $N^2$ randomly chosen local transformations
 $\exp \left\{ i\sum_{i\in \mathrm{local}}\boldsymbol{\alpha }_{i}\widehat{g}_{i}^{N}\right\} $  we generate  $\sim N^{4}$ new product states. We also generate new entangled states  by applying random generic transformations $\exp \left\{i\sum_{i=1}^{N^{2}-1}\boldsymbol{\beta }_{i}\widehat{g}_{i}^{N}\right\}$ to each of the entangled states obtained at the former step. Here  $i$ numerates generators of the $SU(N)$ group while $i\in \mathrm{local}$ mean generators of the subgroup of local transformation. Random parameters and are normally distributed with width gradually decreasing with the number of the iteration step.
We again solve the linear
programming problem for $\sim N^4$ vertices at these two new polytopes and iteratively repeat all the procedure till the result
converges. Note that each next step, the presence of the solution of the former step of the loop is essential
in order to guarantee an outcome  from the linear programming algorithm. The set of the eigenvectors of the density matrix plays this role for the first step.
 Numerical inspection shows that the final results of the algorithm i.e., the product component  $\widehat{\rho }_{\mathrm{sep}}$ and the essentially entangled part  $\widehat{\rho }_{\mathrm{ent}}$,  Eqs.(\ref{sep})-(\ref{ent}), are always the same for different runs. 

The algorithm described above concerns the case of full  separability of a state or else, the identification of  the essentially $K$-entangled component. The same steps, can be applied if we make a repartition of the initial system and consider   $L$-separability  of the state with $L<K$.  Furthermore, if the set of separable states is enlarged to include other special classes of pure states e.g.  states of the $W$ class \cite{Dur}, then one can apply the idea of the algorithm for revealing the  classification  of mixed multipartite entangled state as the one introduced in \cite{SanperaII} for three qubits. We would like to add here that  for the specific  case of three qubits in mixed state, a lot of progress has been recently made on the classification of entanglement via analytic criteria and  efficient algorithms \cite{Lohmayer}-\cite{Rodriques}.

 Finally it is important to mention that linear programming scales polynomially with the dimension of the vector space under consideration in the general case but not always --still a zero-measure of non-polynomial cases may exist.  In consequence, the same additional `rule' has to be applied to the proposed algorithm and the identification of the special cases where the algorithm becomes non-polynomial is an interesting open problem, not resolved in this  work. However, on a practical level  even  in this case, a small random variation of the initial density matrix  brings the problem back to a polynomially complexity.

\section{Suggestions for Characterizing  Entanglement Properties of the essentially entangled component}\label{C}

We may claim that all  information relevant to   entanglement   is  contained in the
essentially entangled part $\widehat{\rho }_{\mathrm{ent}}$ of the density matrix. Though this is not the main object of this work,  we make some simple suggestions for analyzing  entanglement properties of $\widehat{\rho }_{\mathrm{ent}}$ employing previous results \cite{MASV} about characterization of  entanglement for pure states.

For pure quantum states, entanglement is directly related to the factorizability at state vectors, and therefore one can characterize  entanglement by identifying the orbit of local transformations
for a given state. This  orbit can be marked by a complete set of polynomial invariants or alternatively
by the coefficients $\left\{ \beta \right\}$ of the tanglemeter $\widehat{Nl}\left( \left\{ \beta \right\} \right) =\sum_{i,\ldots ,j}\beta
_{i,\ldots ,j}\sigma _{i}^{+}\ldots \sigma _{j}^{+}$ of a given state \cite{MASV}.  The  state defined as $
\left\vert c\left( \left\{ \beta \right\} \right) \right\rangle =\exp \left[\widehat{Nl}\left( \left\{ \beta \right\} \right) \right] \left\vert
0\right\rangle $ is the so called  canonical state, and this can be reached from the state under study by the action of local operations under the constraint that the population of the reference state $\left\vert
0\right\rangle $ is maximized.
In addition to the identification of the orbit of local transformations the tanglemeter generalizes the concept of logarithm to
 vectors and  its coefficients straightforwardly reveal the factorization properties of the state. 
 
Entanglement of mixed states cannot rely only on one operation of group multiplication  but  it should also involve the procedure of casting in convex sums. Therefore the algebraic structure does not suggest a natural framework for the characterization of entanglement in this case. Construction of an approach to entanglement characterization is a convenience just complementing the exhaustive information contained in the essentially entangled part of the density matrix.

A straightforward way to characterize  entanglement of mixed states
would be to find the tanglemeters of the eigenstates of $\rho_{\mathrm{ent}}$. However, it does not mean that an entangled state
corresponding to another orbit cannot be detected. In fact, any pure state which
belongs to the essentially entangled subspace $H_E$ spanned by the eigenvectors of 
$\widehat{\rho }_{\mathrm{ent}}$, is also a legitimate representative of the
ensemble of entangled states associated with this density matrix. One
therefore may want to find the ``corners'' of this ensemble of states by
identifying the state $\left\vert c_{1}\right\rangle $ in  $H_E$
closest to the set of product states $\mathcal{P}$, followed by
identification of a state $\left\vert c_{2}\right\rangle \perp \left\vert
c_{1}\right\rangle $ closest to $\mathcal{P}$ then, $\left\vert
c_{3}\right\rangle \perp \left\vert c_{2}\right\rangle ,\left\vert
c_{1}\right\rangle $ \textit{etc}, till $\left\vert c_{d_E}\right\rangle $, and calculate tanglemeters
for these ``corners''. Tanglemeter coefficients of any state in $H_E$
will therefore be within the borders given by these ``corners''. We would like to mention here that the use of the tanglemeter as a method for characterizing multipartite entanglement is not essential here. One may apply this idea to  other  measures  of multipartite entanglement for pure states as are the entanglement monotones from anti-linear operators introduced in \cite{Osterloh}.

One more option is to find the tanglemeter coefficients distribution function 
\begin{eqnarray}
P\left( \left\{ \beta \right\} \right)  =&\int \left\langle c\left( \left\{
\beta (x)\right\} \right) \right\vert \widehat{\rho }_{\mathrm{ent}
}\left\vert c\left( \left\{ \beta (x)\right\} \right) \right\rangle \nonumber  \\
&\delta \left( \left\{ \beta (x)-\beta \right\} \right) d\mu _{x\in 
H_E}
\end{eqnarray}
resulting from averaging over the Haar measure $\mu _{x\in 
H_E}$ in the subspace $H_E$ in according with the Eq.(\ref{t2}) with  weight
suggested by $\widehat{\rho }_{\mathrm{ent}}$,  the probability to have canonic state with given
tanglemeter coeficients. The number $P\left( \left\{
\beta \right\} \right) $  gives the probability density to find an
entangled state which belongs to the orbit characterized by the set $\left\{
\beta \right\} $ of the tanglemeter coefficients. In the case where one of
the eigenvalues of $\widehat{\rho }_{\mathrm{ent}}$ is much larger than
others, the probability distribution $P\left( \left\{ \beta \right\} \right) 
$ locates near the tanglemeter of the corresponding eigenvector and it can
be adequately characterized by a small covariance matrix of the tanglemeter's
coefficients.

\section{Example}\label{D}
 We now present an illustration  of the introduced methods at a physical example of an open system experiencing loss of coherence due to presence of  classical noise. The model   comprises three elements: two two-level systems and a three-level system. 
The local symmetry group for each of two-level systems is the  $SU(2)$ group, for   the three-level 
 $SU(3)$ group, while for the total assembly the group of transformations (local and non-local) is the $SU(12)$ . We consider the following physical ingredients of the combined system:  an 
atom in $p$-state ($L=1$, $M_{L}=0,\pm 1$) in a static magnetic field, which
parametrically interacts with a two-mode electromagnetic field. We also assume that each of the field modes  allows for two possible polarizations of the photons.

The Hamiltonian of the system  consists of four parts: \\
(i) the Hamiltonian of the first field mode  $\widehat{H}_{1}=k_{z}\left(\widehat{a}_{x}^{\dagger }\widehat{a}_{x}+\widehat{a}_{y}^{\dagger }\widehat{a}_{y}\right) $ with  wavevector $k_{z}$ and polarizations $x$ and $y$,\\
(ii) the Hamiltonian of the second mode $\widehat{H}_{2}=k_{x}\left(\widehat{b}_{y}^{\dagger }\widehat{b}_{y}+\widehat{b}_{z}^{\dagger }\widehat{b}_{z}\right) $ with  wavevector $k_{x}$, 
\\ (iii) the Hamiltonian of the atom $\widehat{H}_{3}=\left( \mathbf{H}\widehat{\mathbf{L}}\right) $ in the static magnetic $\mathbf{H=}\left\{H_{x},H_{y},H_{z}\right\} $ field, where $\widehat{\mathbf{L}}$ the angular momentum vector operator, and \\
(iv) the  Hamiltonian describing the parametric interaction
\begin{equation}
\widehat{H}_{4}=\frac{\left( \widehat{a}_{x}^{\dagger }\widehat{a}_{y}+%
\widehat{a}_{y}^{\dagger }\widehat{a}_{x}\right) \widehat{X}\widehat{Y}}{%
k_{z}-\omega _{1}}+\frac{\left( \widehat{b}_{y}^{\dagger }\widehat{b}_{z}+%
\widehat{b}_{z}^{\dagger }\widehat{b}_{y}\right) \widehat{Y}\widehat{Z}}{%
k_{x}-\omega _{2}},  \label{EQ8}
\end{equation}
which results  from the second order perturbation theory applied over the dipole
interaction $\left( \widehat{a}_{x}^{\dagger }+\widehat{a}_{x}\right) 
\widehat{X}+\left( \widehat{a}_{y}^{\dagger }+\widehat{a}_{y}\right) 
\widehat{Y}+\left( \widehat{b}_{z}^{\dagger }+\widehat{b}_{y}\right) 
\widehat{Y}+\left( \widehat{b}_{z}^{\dagger }+\widehat{b}_{z}\right) 
\widehat{Z}$.\\
Here $\widehat{a}_{i}^{\dagger }$ and $\widehat{b}_{i}^{\dagger }\ $\ are the
photon creation operators of the first and the second mode, respectively,
corresponding to polarization along the direction $i$, while $\widehat{a}_{i}
$ and $\widehat{b}_{i}\ $\ are their conjugate photon annihilation
operators, respectively.\ By $\omega _{1}$ and $\omega _{2}$ we denote 
the frequencies of the allowed dipole atomic transition from the state $p$,
that are closest to the frequencies of the first $k_{z}$ and the second $k_{x}$ photon modes, respectively.\ The atomic optical electron
radius-vector operator $\widehat{\mathbf{R}}=\left\{ \widehat{X},\widehat{Y},
\widehat{Z}\right\} $ and the angular momentum vector operator $\widehat{\mathbf{L}}=\left\{ \widehat{L}_{x},\widehat{L}_{y},\widehat{L}_{z}\right\} $
enter the Hamiltonian as the tensor product and the scalar products with the
magnetic field, respectively, while the light velocity, the electron
charge, and the Planck's constant\ are set to unity.

Since parametric interaction implies conservation of the total number of photons on the two 
 modes, $\widehat{H}_{1}+\widehat{H}_{2}$ is an integral of motion for the system and only the Hamiltonians 
$\widehat{H}_{3}$ and $\widehat{H}_{4}$ are responsible for the dynamical process of interest.
The relevant part $\widehat{H}=\widehat{H}_{3}+\widehat{H}_{4}$  can be re-written
in a more convenient way, noting that  the $x$, $y$, and $z$ components of the vector-operator $\widehat{\mathbf{L}}
$ form an $su(2)$ subalgebra of the symmetry algebra $su(3)$ of the atomic
triplet $p$, while the operators $\widehat{X}\widehat{Y}$ and $\widehat{Y}\widehat{Z}$ entering $\widehat{H}_{4}$ as the tensor product of the
components of $\widehat{\mathbf{R}}$ do not belong to this subalgebra and
yield other generators of $SU(3)$ group. All these operators can be expressed in
terms of Gell-Mann matrices $\widehat{\lambda }_{i}$ with $i=1,\ldots ,8$.
Moreover, the properly normalized bi-lineal photon operators $\widehat{a}_{x}^{\dagger }\widehat{a}_{y}+\widehat{a}_{y}^{\dagger }\widehat{a}_{x}$, $\widehat{a}_{x}^{\dagger }\widehat{a}_{y}-\widehat{a}_{y}^{\dagger }\widehat{a}_{x}$, and $\widehat{a}_{x}^{\dagger }\widehat{a}_{x}-\widehat{a}_{y}^{\dagger }\widehat{a}_{y}$ of the first mode form an $su(2)$ algebra, and
so do the similar operators of the second mode. Therefore, these can be expressed in terms
of the Pauli matrices $\widehat{\sigma }_{1,i}$ and $\widehat{\sigma }_{2,i}$, respectively, with $i=x,y,z$. Summarizing, the  Hamiltonian $\widehat{H}=\widehat{H}_{3}+\widehat{H}_{4}$ can be caste
in the form
\begin{equation}
\widehat{H}=\sum_{i=1}^{3}\widehat{\lambda }_{i}f_{i}+f_{4}\widehat{\sigma }%
_{1,x}\widehat{\lambda }_{4}+f_{6}\widehat{\sigma }_{2,x}\widehat{\lambda }%
_{6}+\varepsilon _{1}\widehat{\sigma }_{1,z}+\varepsilon _{2}\widehat{\sigma 
}_{2,z},  \label{EQ9}
\end{equation}
where the parameters $f_{i=1,2,3}$ depend on the static field, parameters $f_{4}$
and $f_{6}$ are governed by the detuning of the photon frequencies from the
atomic transition frequencies, and parameters $\varepsilon _{i=1,2}$
deviate from zero when the photon frequency turns to be dependent on the
polarization in the presence of an anisotropicity of the refraction index
(that is when $k_{z}$ is slightly different for the $x$ and $y$
polarizations, and similar for $k_{x}$).

Now let us consider a realistic situation where the static field  experiences  small and rapid
fluctuations, that is  $f_{i}(t)=\overline{f}_{i}+\delta f_{i}(t)$ for $i=1,2,3$. In this case  the Liouville equation $i\overset{\cdot }{\widehat{\rho }}=\left[\widehat{H}(t),\widehat{\rho }\right] $ describing the time evolution of the density matrix $\widehat{\rho }(t)$ of the assembly, can be averaged over these rapid fluctuations  $\delta f_{i}(t)$,  yielding \cite{Akulin} the following Lindblad master equation 
\begin{equation}
i\overset{\cdot }{\widehat{\rho }}=\left[ \widehat{\overline{H}},\widehat{\rho }\right] -i\sum_{i,j=1}^{3}\overline{\delta f_{i}(t)\delta f_{j}(t)}\left[ \widehat{\lambda }_{i},\left[ \widehat{\lambda }_{j},\widehat{\rho }\right] \right] ,  \label{EQ1d}
\end{equation}
where the upper bar denotes time average. Substitution to this master
equation in the Liouville  representation \begin{equation} \widehat{\rho }(t)=\sum_{i=0}^{143}%
r_{i}(t)\widehat{g}_{i}^{12}\label{Re}\end{equation} of the density matrix in terms of the
generators of the unitary group $SU(12)$, yields a system
of $143$ linear, first-order differential equations 
\begin{eqnarray}
i\overset{\cdot }{r}_{k} &=&\sum_{m=1}^{143}\left( \mathrm{Tr}\left\{ 1
\widehat{g}_{k}^{12}\left[ \widehat{\overline{H}},\widehat{g}_{m}^{12}\right]
\right\} -i\mathcal{R}_{k,m}\right) r_{m}  \label{EQ2d} \\
\mathcal{R}_{k,m} &=&\sum_{i,j=1}^{3}\overline{\delta f_{i}(t)\delta f_{j}(t)%
}\mathrm{Tr}\left\{ \widehat{g}_{k}^{12}\left[ \widehat{\lambda }_{i},\left[\widehat{\lambda }_{j},\widehat{g}_{m}^{12}\right] \right] \right\}   \notag
\end{eqnarray}
for the real vector components $r_{i}(t)$.
Straightforward analytic solution of  Eq.(\ref{EQ2d}) gives oscillations with time for
some of the coefficients $r_{i}(t)$  while others   die off with  rates determined by the
relaxation operator $\mathcal{R}_{k,m}$.

A considerable amount of work on the understanding of the dynamics of entanglement has been performed so far  and we refer an interested reader to  \cite{Davidovich} for a complete review and reference list. In  Fig.~\ref{fig2} we graphically represent a generic solution for this example, as a spiral in the Liouville space, gradually approaching a stationary solution.   This picture also provides a complementary point of view on the phenomenon  of sudden death and revival of entanglement \cite{Eberly}. With the course of time, we expect the essentially entangled part to oscillate between different subspaces and eventually to vanish for sometime -- when the density matrix is passing inside the polytope of separable states as it is illustrated in Fig.~\ref{fig2}. The revival of  entanglement is marked by the exit  of the density matrix from the polytope. This graphical representation is justified by the calculations which we present in the following.
\begin{figure}[h] \includegraphics[width=0.48\textwidth]{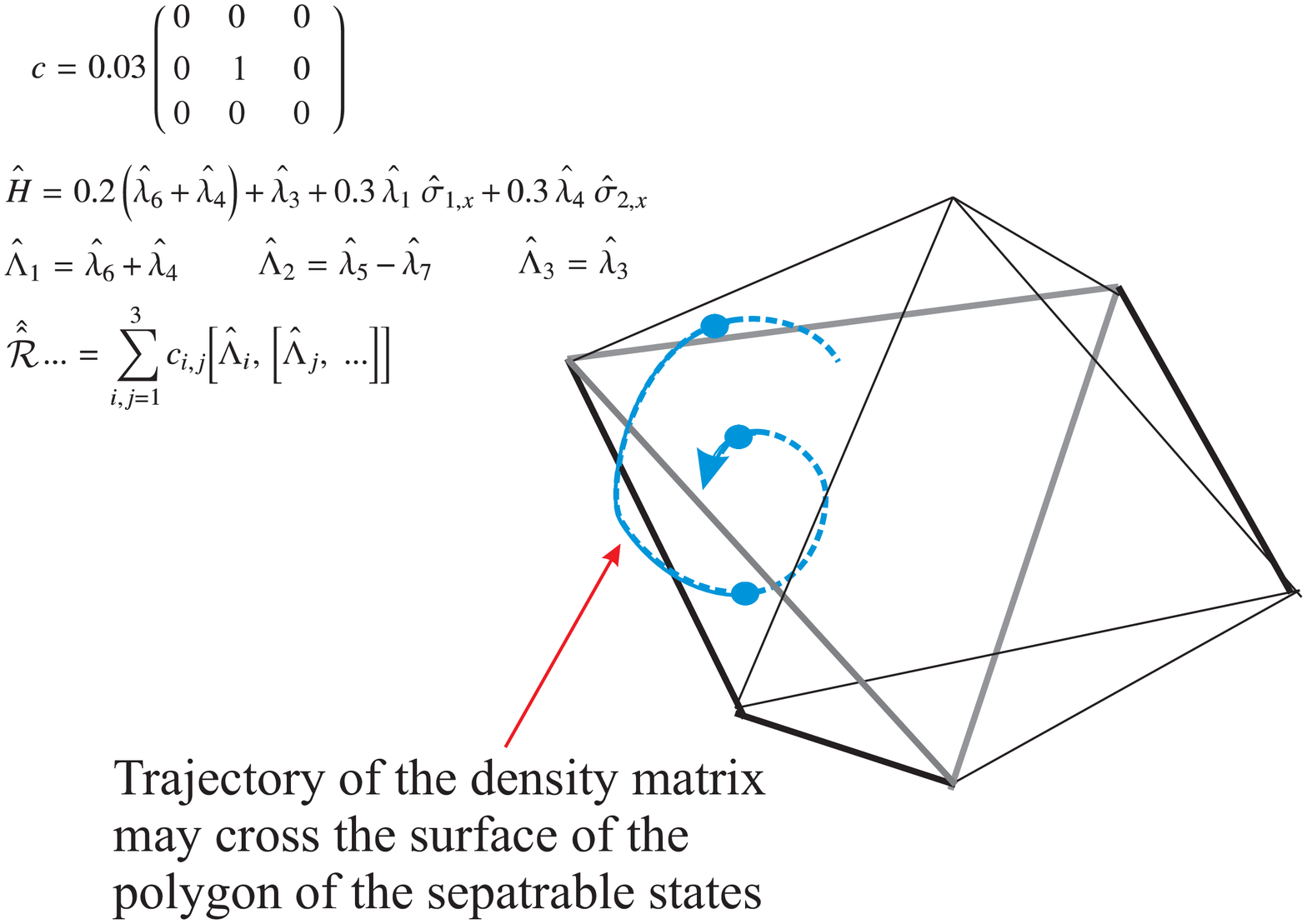}
\vspace{1cm}
\caption{A symbolic description of the trajectory in the Liouville  space of a mixed state  undergoing loss of coherence due to interaction with the environment. Crossing of the polytope of the separable states results in sudden death (or birth) of  entanglement. In the inlet we list the numerical values of the parameters of the model. } \label{fig2}\end{figure}

We now solve  the model Eq.(\ref{EQ2d}) for a  set of given values for  ${f}_{i}$ presented in Fig.~\ref{fig2}, and  reconstruct
the density matrix $\widehat{\rho }(t)$ with the help of Eq.(\ref{Re}). We  summarize the results of calculations  in Fig.~\ref{fig3}.
In Fig.~\ref{fig3}~(a) we plot the purity $P(t)=\mathrm{Tr}\left[\widehat{\rho }^2(t)\right]$ of the density matrix  as a function of time.
At each time step we apply the algorithm and we decompose the density matrix as $\widehat{\rho }(t)=(1-B(t))\widehat{\rho }_{\mathrm{sep}}(t)+B(t)\widehat{\rho }_{\mathrm{ent}}(t)$, Eq.(\ref{one}).
In Fig.~\ref{fig3}~(b) we plot the weight $B(t) =\sum_{i=1}^{m}b_{i}$, Eq.(\ref{EQ2}), of the essentially entangled component in the density matrix. The weight  $B(t)$ is decreasing with time  faster than the purity does and in addition it exhibits some oscillatory behavior that can be probably explained by motion of the essentially entangled component along the facets of the polytope.
In Fig.~\ref{fig3}~(c) we plot the rank $d_{E}$ of $\widehat{\rho }_{\mathrm{ent}}(t)$,  and we observe that this moves in a rather random way between the value $1$ and $5$. We note that if full ($K=3$) separability is considered then $d_{E\max}=7$. However in our program we have included in the ``polytope'' of separable states also the bi-separable states, thus actually  $d_{E\max}=5$.   The ``jumps'' of the rank   demonstrate  the recursive move of essentially entangled component  between different essentially entangled subspaces on $S_{CH}$.  Moreover, in the time interval $~[18.8-19.7]$, $B(t)$ vanishes implying that the state enters inside the polytope of separable states. This physical situation describes a sudden death and sudden revival of entanglement a  phenomenon \cite{Eberly}-\cite{Lopez} which has been studied extensively with other methods. Our geometric decomposition offers additional information on the origin of this phenomenon, see Fig.~\ref{fig2}.

In order to analyze the entanglement properties of the essentially entangled component we first note that for the chosen model system in the vast majority of the time steps, there is a dominant eigenvector $\widehat{e}_{\mathrm{dom}}$ for $\widehat{\rho }_{\mathrm{ent}}$ with a corresponding eigenvalue $\lambda_{dom}>0.9$, see  Fig.~\ref{fig3}~(d). Therefore, for this specific example and assigned parameters,  it makes sense just to analyze  entanglement properties of $\widehat{e}_{\mathrm{dom}}$, whenever the condition $\lambda_{dom}>0.9$ is satisfied, and  to conclude from this analysis the entanglement properties of $\widehat{\rho }_{\mathrm{ent}}$. Naturally, this analysis together with the weight $B(t)$, give all the information necessary to describe  entanglement in  $\widehat{\rho}$. 

\begin{center}
\begin{figure*}[t]{\centering{\includegraphics*[width=0.8\textwidth]{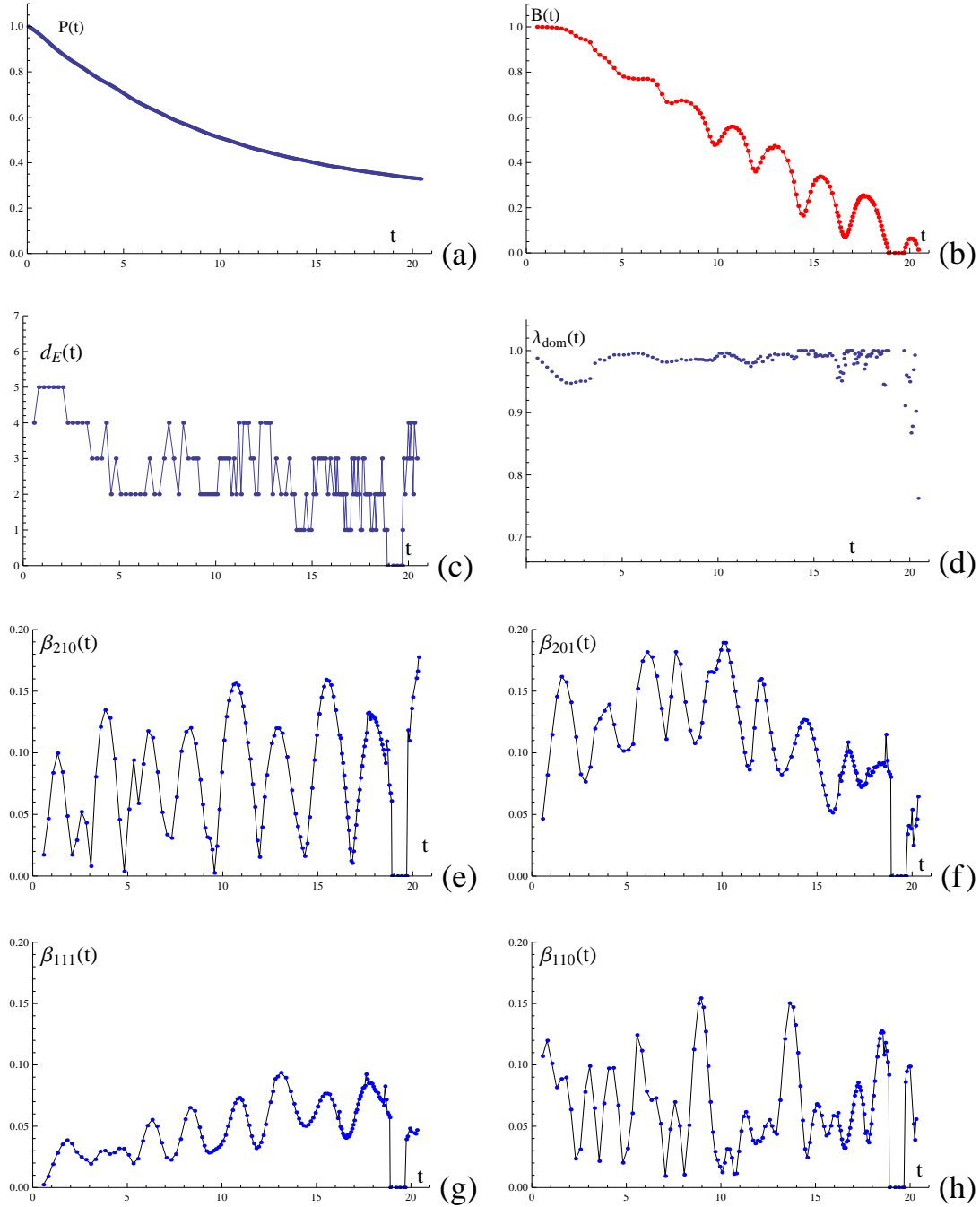}}}\caption{We solve the Lindbland equation for the example and we apply the algorithm at each time step. (a) Purity of the assembly. (b) The statistical contribition of $\widehat{\rho }_{\mathrm{ent}}(t)$ to the density matrix. (c) The rank of $\widehat{\rho }_{\mathrm{ent}}(t)$.  (d) The eigenvalue of the dominant eigenvector of $\widehat{\rho }_{\mathrm{ent}}(t)$.  (e)-(h) The oscillations of the real coefficients of the tanglemeter. In the time interval $~[18.8-19.7]$ sudden death of entanglement takes place and then its revival. } \label{fig3}\end{figure*}
\end{center}

We analyze the entanglement properties of  $\widehat{e}_{\mathrm{dom}}$ with the help of the method of nilpotent polynomials \cite{MASV}. In the Appendix we provide an explicit method for deriving the general expression for the tanglemeter of a wavevector describing an assembly of a three-level system and two two-level systems: 
\begin{eqnarray}\widehat{Nl}\left(\left\{ \beta \right\} \right) &=&\left( \beta _{110}\hat{t}^{+}\hat{\sigma}_{1}^{+}+\beta _{101}\hat{t}^{+}\hat{\sigma}_{2}^{+}+\beta _{011}\hat{\sigma}_{1}^{+}\hat{\sigma}_{2}^{+}+\right. \nonumber \\&&\left. \beta _{210}\hat{u}^{+}\hat{\sigma}_{1}^{+}+\beta_{201}\hat{u}^{+}\hat{\sigma}_{2}^{+}+\right. \nonumber \\&&\left.\beta _{111}\hat{t}^{+}\hat{\sigma}_{1}^{+}\hat{\sigma}_{2}^{+}\right)  \end{eqnarray}
with $\beta _{111},\beta _{201},\beta _{210},\beta _{110}$ being positive numbers and $\beta _{101},\beta _{011}$ being complex. The matrix representation of the nilpotent variables (operators) $\hat{u}^+$, $ \hat{v}^+$, $\hat{\sigma}^+$ is also provided in the Appendix. Concerning now the physical meaning of the coefficients. The coefficients of the tanglemeter even though are not entanglement monotones \cite{Horo} in the strict sense, these are invariant under the action of local transformations and  the presence of any non-zero term in the tanglemeter ensures the presence of entanglement. More precisely, the coefficient $\beta _{111}$ ensures the presence of genuine tripartite entanglement in the state while the rest of the coefficients are related to bipartite entanglement.  
In Fig.~\ref{fig3}~(e)-(h) we plot  those coefficients which are positive, and we observe that these oscillate with time without dissipation. The same holds for the real and imaginary parts of  the complex ones not shown on the figure.

With this example, in addition to the death and revival of entanglement,  we observe two interesting phenomena which need more case study in order to decide whether are specific to this example or general. The first is the presence of a dominant eigenvector in the essential entangled component and the second is the oscillations without dissipation of the entanglement characteristics of the essential entangled component.

\section{Conclusions}\label{E}

In this work we have studied a  concept related to entanglement of mixed states namely the essentially entangled component of a mixed multipartite state and more important, we have suggested an efficient algorithm for its identification. The essentially entangled component  is the complementary part to the best separable approximation introduced in \cite{Sanpera} and this naturally contains all the entanglement of the density matrix. We analyze some properties of the essentially entangled components and we suggest methods for characterizing its entanglement content.

 Our main  tool is the accustomed geometric description of mixed quantum states in the spirit of Bloch vector representation, which results from the decomposition of a density matrix over the generators of the relevant group. We have shown that pure states are not everywhere on the surface of this hypersphere, in the contract to the  Bloch vector, and that the convex hull of pure states from a convex ``body'' inside the sphere. The convex hull  of separable states forms a convex ``polytope'' inside the ``body'' of general states. As a consequence the entangled states inside the body and outside the polytope can be represented as sum of a separable state on the surface of the polytope and an essentially entangled component located on the surface of the ``body''.  This geometric picture gives the guidance  for  constructing the algorithm  and for analyzing the properties of the essentially entangled component. The latter being located on the surface of the ``body'', form there sets of lower dimensions, such that the rank of the relevant density matrix does not exceed a number which depends on the dimensions of the total system,  and on its chosen partition.

Finally, at a particular example we study the dynamics of an open quantum system and we reconstruct the time trajectory of the decomposed density matrix inside the convex ``body''.  Sudden death and sudden birth of entanglement can be seen as the results of crossing of the of the trajectory of the density matrix with the surface of the ``polytope'' of separable states. There are some other interesting phenomena  appearing in this example but these still need further studies to lead to general conclusions.

Concerning possible applications of the results. The algorithm introduced in this work scales polynomially with the dimension of the system in the general case, and it can be employed to study  open questions about entanglement in mixed states. For instance, this can be applied straightforwardly to address the question of the relative volume of separable states over entangled mixed states as function of the total purity of the system and the total dimension of the system \cite{Volume}. An answer to this question  can serve to the evaluation of emerging quantum technologies and their quantum limits. Moreover, the essentially entangled component containing all entanglement properties of the density matrix may also  provide new directions to entanglement detection \cite{Guhne} and entanglement distillation \cite{Bennett} techniques.

\section*{Acknowledgement}
 VA acknowledges stimulating and useful discussions with Mikhail Tsvasman and Sergey Pirogov‏ and the hospitality accorded to him at Laboratoire J.-V. Poncelet CNRS. VA and AM are thankful to Jens Siewert for indicating to them  important related works. AM acknowledges  financial support from the  Ministry of Education and Science of the Republic of Kazakhstan (Contract $\#$ $339/76-2015$).
\renewcommand{\theequation}{A-\arabic{equation}}
  
  \setcounter{equation}{0} 
\section*{APPENDIX}
\subsection{ A second formulation and proof of the main theorem}

Here we provide a more detailed formulation and proof for the  theorem given in Section~\ref{A} which does not relay on a particular quantum mechanical representation.

\textit{The maximum rank $d_{E\max }$ of an essentially entangled component
  is $N_{CG}-N_{CS}$, where   $N_{CG}$ is the
dimension of the Cartan subgroup of the group of all transformations on the state and $N_{CS}$ is the dimension of  Cartan (sub)subgroup  generating only local transformations.}

\textsl{Remark.} The numbers $N_{CG}$ and $N_{CS}$ give the numbers of complex parameters characterizing generic and product state vectors, respectively, on $N=N_{CG}+1$ dimensional Hilbert spaces.
\begin{proof}
Consider a density matrix $\widehat{\rho }$  and its decomposition to the essentially entangled and separable part $\widehat{\rho }= (1-B) \widehat{\rho }_{\mathrm{sep}}+ B \widehat{\rho }_{\mathrm{ent}}$. Since $B$ corresponds to a minimum value of all possible weights,  we conclude that no $\epsilon>0$ and product vector $\left|p\right\rangle$ exist such that $ \widehat{\rho }_{\mathrm{ent}}-\epsilon\left|p\right\rangle\left\langle p\right|$ is a positive matrix. Considering now the essentially entangled subspace $H_E$ spanned by the eigenvectors $\left|\psi_{i}\right\rangle$ with non-zero eigenvalues of $\widehat{\rho}_{\mathrm{ent}}$ with $i=1,\ldots,d_E$, this condition means that no product state  $\left|p\right\rangle$ exists in $H_E$. Indeed, for the case where
\begin{equation}
\left|p\right\rangle=\sum_{i=1}^{d_E}\left|\psi_{i}\right\rangle+ \epsilon' \left|p'\right\rangle
\end{equation}
with $\left\langle p'\right|\left.\psi_{i}\right\rangle=0$ for every $i=1,\ldots,d_E$ one identifies the vector $\left|p'\right\rangle$  orthogonal to the subspace of $d_E$  eigenvectors which makes
 \begin{eqnarray}
\left\langle p'\right|\left(\widehat{\rho }_{\mathrm{ent}}-\epsilon\left|p\right\rangle\left\langle p\right|\right)\left|p'\right\rangle
\nonumber \\=-\epsilon\left|\left\langle p\right|\left.p'\right\rangle\right|^2<0\;,
\end{eqnarray}
and therefore extremality implies that no product state is orthogonal to the orthogonal compliment $H_E^\bot$ of $H_E$ spanned by the eigenvectors  $\left|\psi_{i}\right\rangle$ of $\widehat{\rho}_{\mathrm{ent}}$ with zero eigenvalues and $i=d_E,\ldots,N_{CG}$.

In other words, in order to find such a state we have to satisfy $N_{CG}-d_E$ equations $\left\langle p\right|\left.\psi_i\right\rangle=0$ with $i=d_E+1,\ldots,N_{CG}$ for a product state $\left|p\right\rangle$ given by specification of its $N_{CS}$ parameters. This is impossible when $N_{CG}-d_E\geq N_{CS}$, which determines the maximum rank $d_{E\max}$ of $\widehat{\rho}_{\mathrm{ent}}$.
\end{proof}


\subsection{ Deriving the tanglemeter of the physical example in Section~\ref{D}}
The system  under consideration consists of the two modes of the field interacting with a  three-level atom. The Hilbert space thus is of  dimension $N=12$,
  a direct product of the spaces of  two two-level systems (qubits) and of one three level system (qutrit).
In the standard computational basis a state vector of the system is  expressed  as
\begin{eqnarray*}
\left\vert \Psi \right\rangle  &=&\psi _{000}\left\vert 000\right\rangle+\psi _{100}\left\vert 100\right\rangle +\psi _{200}\left\vert200\right\rangle +\psi _{010}\left\vert 010\right\rangle+  \\&&  \psi_{001}\left\vert 001\right\rangle\psi _{110}\left\vert 110\right\rangle +\psi_{101}\left\vert101\right\rangle +\psi _{011}\left\vert 011\right\rangle + \\&&\psi _{210}\left\vert 210\right\rangle +\psi _{201}\left\vert201\right\rangle +\psi _{111}\left\vert 111\right\rangle +\psi_{211}\left\vert 211\right\rangle .\end{eqnarray*}
or alternatively using the nilpotent creation operators
\begin{equation} \hat{u}^{+}=\begin{pmatrix} 0 & 0 & 1 \\ 0 & 0 & 0\\  0 & 0 & 0 \end{pmatrix}, \end{equation}
\begin{equation}\hat{t}^{+}=\begin{pmatrix} 0 & 0 & 0 \\ 0 & 0 & 1\\  0 & 0 & 0 \end{pmatrix}, \end{equation}
\begin{equation}\hat{\sigma}^{+}=\begin{pmatrix} 0 &  1 \\ 0 & 0  \end{pmatrix}\end{equation}

as  \begin{eqnarray*}\left\vert \Psi \right\rangle  &=\left( \psi _{000}+\psi _{100}\hat{t}^{+}+\psi _{200}\hat{u}^{+}+\psi _{010}\hat{\sigma}_{1}^{+}+\psi _{001}\hat{\sigma}_{2}^{+}\right.  \\ &\psi _{110}\hat{t}^{+}\hat{\sigma}_{1}^{+}+\psi _{101}\hat{t}^{+}\hat{\sigma}_{2}^{+}+\psi _{011}\hat{\sigma}_{1}^{+}\hat{\sigma}_{2}^{+}+\psi _{210}\hat{u}^{+}\hat{\sigma}_{1}^{+}+ \\&\left. \psi _{201}\hat{u}^{+}\hat{\sigma}_{2}^{+}+\psi _{111}\hat{t}^{+}\hat{\sigma}_{1}^{+}\hat{\sigma}_{2}^{+}+\psi _{211}\hat{u}^{+}\hat{\sigma}_{1}^{+}\hat{\sigma}_{2}^{+}\right) \left\vert 000\right\rangle .\end{eqnarray*}

The next step that should be performed is the application  of all the available local transformations ($SU(3)\otimes 1\otimes 1$, $1\otimes SU(2)\otimes 1$, $1\otimes 1\otimes SU(2)$) on the given state $\left\vert \Psi \right\rangle$ in order to  construct the corresponding canonic state $\left\vert \Psi_{c}\right\rangle $ which marks the orbit of local transformations.  To simplify the procedure, we apply the local transformations  on a given $\left\vert \Psi \right\rangle$ in the following order:

\textsl{(a)} We first apply local operations  generated by the operators $\left\{ \hat{\sigma}_{1}^{x},\hat{\sigma}_{1}^{y},\hat{\sigma}_{2}^{x},\hat{\sigma}_{2}^{y},\hat{\lambda}_{4},\hat{\lambda}_{5},\hat{\lambda}_{6},\hat{\lambda}_{7}\right\} $ and we require that the  polulation of the reference level $\left\vert 000\right\rangle$ is getting maximum. Under this condition the populations of the levels : $\left\vert 100\right\rangle,\left\vert 200\right\rangle,\left\vert 010\right\rangle,\left\vert 001\right\rangle$ are vanishing.

\textsl{(b)} We then apply local operations generated by $\left\{ \hat{\lambda}_{1},\hat{\lambda}_{2}\right\} $ to maximize also the population of the level $\left\vert 111\right\rangle $. This way the contribution of the level $\left\vert 211\right\rangle $  also vanishes.

\textsl{(c)} Finally we apply local operations generated by $\left\{ \hat{\sigma}_{1}^{z},\hat{\sigma}_{2}^{z},\hat{\lambda}_{3},\hat{\lambda}_{8}\right\} $ in order to make the   phase of $\left\vert 111\right\rangle,\left\vert 210\right\rangle,\left\vert 201\right\rangle,\left\vert 110\right\rangle$ equal to  the phase of the amplitude of the reference level $\left\vert 000\right\rangle$.

After this procedure one obtains the following form for the unormalized canonic state:
\begin{eqnarray}
\left\vert \Psi _{c}\right\rangle   = &\left( 1+\alpha _{110}\hat{t}^{+}\hat{\sigma}_{1}^{+}+\alpha _{101}\hat{t}^{+}\hat{\sigma}_{2}^{+}\right. \nonumber \\
& +\alpha _{011}\hat{\sigma}_{1}^{+}\hat{\sigma}_{2}^{+}\alpha _{210}\hat{u}^{+}\hat{\sigma}_{1}^{+}+\alpha_{201}\hat{u}^{+}\hat{\sigma}_{2}^{+}
\nonumber \\&\left.+\alpha _{111}\hat{t}^{+}\hat{\sigma}_{1}^{+}\hat{\sigma}_{2}^{+}\right) \left\vert 000\right\rangle \label{f} \end{eqnarray}
with $\alpha_{111},\alpha_{201},\alpha _{210},\alpha _{110}$ being positive numbers and $\alpha _{101},\alpha _{011}$ being complex.

The final step for arriving to the tanglemeter $\widehat{Nl}\left( \left\{ \beta \right\} \right)$ of the state is to take the logarithm of the 
 polynomial on the nilpotent variables $\hat{t}^{+}$, $\hat{\sigma}^{+}_{1,2}$  in Eq.(\ref{f}). 
It is easy to show that 
\begin{eqnarray*}
\widehat{Nl}\left( \left\{ \beta \right\} \right)  & =\beta_{110}\hat{t}^{+}\hat{\sigma}_{1}^{+}+\beta_{101}\hat{t}^{+}\hat{\sigma}_{2}^{+}+\beta_{011}\hat{\sigma}_{1}^{+}\hat{\sigma}_{2}^{+} \nonumber \\
& \beta_{210}\hat{u}^{+}\hat{\sigma}_{1}^{+}+\beta_{201}\hat{u}^{+}\hat{\sigma}_{2}^{+}+\beta_{111}\hat{t}^{+}\hat{\sigma}_{1}^{+}\hat{\sigma}_{2}^{+}\end{eqnarray*}
with $\beta_{110}=\alpha_{110}$, $\beta_{101}=\alpha_{101}$ \textsl{etc}.

\end{document}